\documentclass[12pt, reco]{article}
\usepackage{amssymb,amscd,amsmath,amsthm}
\usepackage[colorinlistoftodos]{todonotes}
\usepackage[colorlinks=true, allcolors=blue]{hyperref}
\usepackage{lmodern}
\usepackage[english]{babel}
\usepackage{color}
\usepackage{graphicx}
\usepackage{tikz}
\usepackage[margin=1in]{geometry}
\newtheorem{theorem}{Theorem}[section]
\newtheorem{proposition}[theorem]{Proposition}

\newtheorem{lemma}[theorem]{Lemma}

\newtheorem{example}[theorem]{Example}
\newtheorem{remark}[theorem]{Remark}
\def\Tr{\mathrm{Tr}}
\def\id{{\bf 1}\!\!{\rm I}}

\title{Matrix Product States as Observations   of Entangled Hidden Markov Models  }

\begin{document}

\maketitle

\centerline{ \author{\Large $^{1,2}$Abdessatar Souissi}}
\vskip0.25cm
\centerline{$^1$ Department of Management Information Systems, College of Business and Economics, }
\centerline{Qassim University, Buraydah 51452, Saudi Arabia}
\centerline{$^2$ Department of Mathematics, Preparatory Institute for Engineering Studies,}
\centerline{ University of Monastir,  Monastir 5000, Tunisia}

\centerline{\textit{a.souaissi@qu.edu.sa}}

\tableofcontents

\begin{abstract}
This paper reveals the intrinsic structure of Matrix Product States (MPS) by establishing their deep connection to entangled hidden Markov models (EHMMs). It is demonstrated that a significant class of MPS can be derived as the outcomes of EHMMs, showcasing their underlying quantum correlations. Additionally, a lower bound is derived for the relative entropy between the EHMM-observation process and the corresponding MPS, providing a quantitative measure of their informational divergence. Conversely, it is shown that every MPS is naturally associated with an EHMM, further highlighting the interplay between these frameworks. These results are supported by illustrative examples from quantum information, emphasizing their importance in understanding entanglement, quantum correlations, and tensor network representations.
\end{abstract}


\section{Introduction}\label{sect_prel}
In quantum information, quantum channels describe how quantum states are transmitted and transformed. Mathematically, they are represented as completely positive trace-preserving (CPTP) maps.   In finite-dimensional settings, these maps are often expressed using the Kraus decomposition \cite{Kraus71}
\[
\Phi(X) = \sum_{j} K_j X K_j^{\dagger}
\]
where the operators \(B_k\) satisfy the isometry condition
\[
\sum_{j} K_j^{\dagger} K_j = \id
\]
Quantum transition expectations are completely positive and identity-preserving (CPIP) maps between C$^*$-algebras. These maps give rise to Markov operators, which offer a dual perspective on quantum channels, with their associated Kraus matrices \( K_j \) satisfying the  following gauge condition
\begin{equation}\label{cpip}
\sum_{j} K_j K_j^{\dagger} = \id
\end{equation}
The algebraic nature of these transition expectations  provides a rigorous framework for understanding the connection between quantum dynamics and their underlying mathematical properties. This framework forms the foundation of Quantum Markov Chains (QMCs), introduced by Accardi in 1974 \cite{Acc75}. QMCs are  states defined on infinite tensor products of matrix algebras, offering a powerful probabilistic tool to analyze correlations and dynamics in quantum systems within the context of quantum information theory. The development of QMCs was motivated by the pursuit of a quantum analogue to Dobrushin's theory of Markov fields \cite{D68} and was further influenced by Araki's seminal work on quantum lattice systems \cite{Araki71, Araki69}. Over the years, QMCs have become a precise and robust tool for modeling the dynamics of one-dimensional quantum systems and advancing quantum information processing \cite{Ib08,   G08, AF83, AcSouElG20, FR15}.  Hidden Quantum Markov Models (HQMMs) \cite{AGLSclass24, AGLSQuan24, MonrWiesn11} provide a framework for modeling quantum systems with underlying  QMCs, offering insights into quantum stochastic processes.

Finitely correlated states (FCS), introduced by Fannes, Nachtergaele, and Werner \cite{FNW1, FNW2, FNW3}, represent a distinguished class of QMCs that have proven instrumental in modeling the ground states of valence bond solid systems, including those described by the AKLT framework \cite{AKLT}. Closely related to FCS, matrix product states (MPS) provide a compact and efficient representation of quantum states, offering profound insights into their structure \cite{perez2007matrix, Mug2012}.
 A MPS  with periodic boundary condition is expressed through a sequence of matrices \(\{A_k\}\),  each corresponding to a local basis state \(|k\rangle\)  of a Hilbert space $\mathcal{K}$  and acting on an auxiliary (bond) space $\mathcal{H}$ of dimension \(m\) and it can be expressed as
\[
|\psi_N\rangle = \sum_{k_1, k_2, \dots, k_N} \mathrm{Tr}\left( A_{k_1}^{[1]} A_{k_2}^{[2]} \cdots A_{k_N}^{[N]} \right) |k_1 k_2 \dots k_N\rangle\in\mathcal{K}^{\otimes N}
\]
Matrix Product States (MPS) are pivotal in determining ground states of one-dimensional quantum many-body systems, particularly through the Density Matrix Renormalization Group (DMRG) algorithm \cite{schollwock2011density, scho05}. Despite their extensive applications in physics and quantum information \cite{perez2007matrix, verstraete2008matrix, orus2014practical, VC06}, the probabilistic interpretations of MPS require further investigation, especially regarding their connections to quantum processes.  Notably, the work in \cite{Mov22} employs an ergodic quantum channels approach to analyze the thermodynamic limit for a specific class of MPS, offering further insights into their structural and dynamical properties.

This paper establishes a direct connection between MPS and Entangled Hidden Markov Models (EHMMs) \cite{SS23}, a special class of HQMMs. EHMMs describe local correlations within a bipartite quantum system $\mathcal{H} \otimes \mathcal{K}$, where the hidden system, represented by a Hilbert space $\mathcal{H}$, evolves according to an entangled Markov chain (EMC) \cite{AF05, AOM06, SSB23}, and the observation system is governed by a Hilbert space $\mathcal{K}$. The correlation between the hidden EMC and the observation process is captured by a conditional independence relation introduced in \cite{AGLSQuan24}, which extends classical hidden Markov models (HMMs) into the quantum domain. This perspective highlights how HQMMs  can naturally lead us to  MPS and vice-versa, enriching the study of quantum systems and their dynamics.

The primary finding of this study Theorem \ref{thm_main1} establishes a significant connection between Matrix Product States (MPS) with periodic boundary conditions (PBC) and entangled hidden Markov models (EHMMs). A key result demonstrates that certain PBC MPS $|\psi_N\rangle \in \mathcal{K}^{\otimes N}$ can be rigorously interpreted as partial observations of an EHMM, represented by the state vectors $\big|\Psi_{H,O;n}\big\rangle \in \mathcal{H}^{\otimes (n+1)} \otimes \mathcal{K}^{\otimes n}$, via partial measurement on a specific  state vector $E_{N,n}$. Explicitly, this relation is expressed as:
\[
\big\langle \Psi_{H,O;n} \mid E_{N,n} \big\rangle_{\mathcal{H}^{\otimes (n+1)} \otimes \mathcal{K}^{\otimes (n-N)}} = |\psi_N\rangle.
\]

This result is framed using the Schrödinger picture for EHMMs, which is shown to be equivalent to the algebraic approach described in \cite{AGLSQuan24, SS23}. In this context, EHMMs are defined as states $\varphi_{H,O}$ on the infinite system
\[
\bigotimes_{n \in \mathbb{N}} \big(\mathcal{B}(\mathcal{H}) \otimes \mathcal{B}(\mathcal{K})\big),
\]
ensuring the existence of a thermodynamic limit:
\[
\varphi_{H,O}(Z) = \lim_{n \to \infty} \big\langle \Psi_{H,O;n} \big| Z \big| \Psi_{H,O;n} \big\rangle.
\]

This thermodynamic limit underscores the presence of FCS for a broad class of Hamiltonians, providing a solid   framework for the MPS under consideration. Additionally, the   sequence of states, \( \varphi_N(\cdot) = \langle \psi_N | \cdot | \psi_N \rangle \), defined on the observable algebra \( \bigotimes_{\mathbb{N}} \mathcal{B}(\mathcal{K}) \), demonstrates profound connections to the entangled Markov chain (EMC) on the hidden algebra \( \bigotimes_{\mathbb{N}} \mathcal{B}(\mathcal{H}) \), which is associated with the auxiliary space of the MPS.

Theorem \ref{thm_main2}, the second main result, establishes that for any Matrix Product State (MPS) whose tensors satisfy the gauge condition (\ref{cpip}), there exists a rigorous and canonical method to construct an associated Entangled Hidden Markov Model (EHMM). This result encompasses a broad class of MPS, including numerous well-established examples in the literature. The class of EHMMs emerging from this construction raises important questions regarding their stochastic equivalence, degrees of entanglement, and their potential equivalence to the EHMMs defined in Theorem \ref{thm_main1}. As a measure of distinguishability between the MPS \(|\psi_N\rangle\) and the observation state \(|\Psi_{O;N}\rangle\), we derive in Theorem \ref{thm_entropyInequality} a lower bound for the relative entropy \cite{Araki75,  U62, OP} of their density matrices
\[
S(\rho_N \| \rho_{O;N}) = \Tr\big(\rho_{N}(\log(\rho_N) - \log(\rho_{O;N}))\big)
\]
 This result quantifies the difference between the two states and paves the way for systematically compressing various measures of entanglement, such as negativity \cite{VW2002} and the degree of entanglement \cite{Ohya} for these states. This topic will be explored in future research.
The results are illustrated through prominent examples from well-known MPS, such as the GHZ,  cluster and AKLT states. These examples highlight the connection between the formalisms of EHMMs and MPS, offering valuable insights into their interplay and potential applications in  quantum information and computing.

\section{Schrödinger Picture of EHMMs}\label{sect-EHMMs}
  Let \( \mathcal{H} \) be a Hilbert space of dimension \( m \), equipped with an orthonormal basis \( e= \{e_i \mid i \in I_H\} \), where the index set is \( I_H = \{1, 2, \ldots, m\} \). The associated space of bounded linear operators on \( \mathcal{H} \), denoted \( \mathcal{B}_H = \mathcal{B}(\mathcal{H}) \), is defined with the identity operator \( \id_H \) and operator norm \( \|\cdot\| \), it represents the hidden algera.

Similarly, let \( \mathcal{K} \) be a \( d \)-dimensional Hilbert space with orthonormal basis \(f = \{|k\rangle \mid k \in I_O\} \), where \( I_O = \{1, 2, \ldots, d\} \). The space of bounded linear operators on \( \mathcal{K} \) is denoted by \( \mathcal{B}_O = \mathcal{B}(\mathcal{K}) \) with identity \(\id_O\), it represents the observation algebra.
A map $\mathcal{E}$ between two C$^*$-algebrass, is completely positive if for every $n$ the map $\mathcal{E}\otimes id_n$ is positive.
A \textbf{quantum channel} is a completely positive trace-preserving (CPTP) map \( \Phi \) that acts on a density matrix \( \rho \) as:
\[
\Phi(\rho) = \sum_i K_i \rho K_i^\dagger
\]
where \( K_i \) are Kraus operators satisfying the isometry condition:
\[
\sum_i K_i^\dagger K_i = \id
\]
This map describes the evolution of quantum states in an open system, preserving the trace and positivity of the density matrix.

A \textbf{quasi-conditional expectation} with respect to a triplet of C$^*$-algebras algebras \( \mathcal{C} \subseteq \mathcal{B} \subseteq \mathcal{A} \) is a completely positive, identity-preserving (CPIP) linear map \( E: \mathcal{A} \to \mathcal{B} \) satisfying the condition:
\begin{equation}\label{eq_QCE}
E(c\, a) = c\,E(a), \quad \forall c \in \mathcal{C}, \; \forall a \in \mathcal{A},
\end{equation}
where \( \mathcal{C} \) is a central subalgebra.

In the context of  QMCs \cite{Acc75, AcSouElG20}, a typical scenario involves \( \mathcal{A} = \mathcal{A}_{\text{past}} \otimes \mathcal{A}_{\text{present}} \otimes \mathcal{A}_{\text{future}} \). Here, a quasi-conditional expectation with respect to the triplet \( \mathcal{A}_{\text{past}} \subseteq \mathcal{A}_{\text{past}} \otimes \mathcal{A}_{\text{present}} \subseteq \mathcal{A} \) can be expressed as \( E = \text{id}_{\mathcal{A}_{\text{past}}} \otimes \mathcal{E} \), where \( \mathcal{E} \) is a CPIP map from \( \mathcal{A}_{\text{present}} \otimes \mathcal{A}_{\text{future}} \) to \( \mathcal{A}_{\text{present}} \), known as the \textbf{transition expectation}. This map encapsulates the Markovian structure of \( E \).


Recall that, the \textbf{Schur product}, denoted by \( X \diamond X' \), is defined as the element-wise product of two operators \( X = \sum_{i,j} x_{ij} e_ie_j^\dagger \) and \( X' = \sum_{i,j} x'_{ij} e_ie_j^\dagger \) in $\mathcal{B}(\mathcal{H})$, and it is expressed as:
\begin{equation}\label{tensor-schur}
    X \diamond X' := \sum_{i,j} x_{ij} x'_{ij} \,  e_ie_j^\dagger \in \mathcal{B}(\mathcal{H})
\end{equation}
For each $n$, we define the operator \( V_{H}^{[n]} : \mathcal{H} \to \mathcal{H} \otimes \mathcal{H} \) by the linear extension of:
\begin{equation}\label{VH}
  V^{[n]}_{H} e_i = \sum_{j} U^{[n]}_{ij} e_i\otimes e_j
\end{equation}
where \( U^{[n]}_{ij} \) satisfies \( \Pi^{[n]}_{ij} = |U^{[n]}_{ij}|^2 \). It is straightforward to verify that \( V^{[n]\,{\dagger}}_{H} V^{[n]}_{H} = id_{\mathcal{H}} \), ensuring that \( V^{[n]}_{H} \) is an isometry. The map \( \mathcal{E}_{H;n} \equiv V^{[n]\, \dagger}_{H} (\cdot) V^{[n]}_{H} \) defines a transition expectation from \( \mathcal{B}_{H;n} \otimes \mathcal{B}_{H;n+1} \) into \( \mathcal{B}_{H;n} \), which is explicitly given on localized observable $X_n = \sum_{i,j}X_{n;ij} e_ie_j^{\dagger}\in\mathcal{B}_{H; n}$, $X_{n+1} = \sum_{i,j}\, X_{n+1;ij}e_ie_j^{\dagger} \in \mathcal{B}_{H; n+1}$   by:
$$
\mathcal{E}_{H;n}(X_n\otimes X_{n+1}) = \sum_{i,j}\sum_{k,l}\overline{U^{[n]}_{ik}}\, U^{[n]}_{jl}\, X_{n;ij}X_{n+1;kl} e_ie_j^{\dagger}
$$
 Its dual map, \( \mathcal{E}_{H;n}^{\dagger}\equiv V^{[n]}_{H} (\cdot) V^{[n]\, \dagger}_{H} \), acts as a quantum channel from \( \mathcal{B}_{H;n}^{\dagger} \) to \( (\mathcal{B}_{H;n} \otimes \mathcal{B}_{H;n+1})^{\dagger} \). Similarly, for each \( m \times d \) stochastic matrix \( Q^{[n]} = \big(Q^{[n]}_{i}(k)\big) \), we define the operator \( V_{O}^{[n]} : \mathcal{H} \to \mathcal{H}\otimes \mathcal{K} \) by the linear extension of:
\begin{equation}\label{VOn}
 V^{[n]}_{O}(e_{i}) = \sum_{k} \chi^{[n]}_{i}(k)\, e_{i} \otimes |k\rangle
\end{equation}
where \( \chi^{[n]}_{i}(k) \) satisfies \( Q^{[n]}_{i}(k) = |\chi^{[n]}_{i}(k)|^2 \). It is straightforward to verify that \( V^{[n]\,{\dagger}}_{O} V^{[n]}_{O} = id_{\mathcal{H}} \), ensuring that \( V^{[n]}_{O} \) is a partial isometry.

The map \( \mathcal{E}_{O;n} \equiv V^{[n]\, \dagger}_{O} (\cdot) V^{[n]}_{O} \) defines an emission expectation from \( \mathcal{B}_{H;n} \otimes \mathcal{B}_{O;n} \) into \( \mathcal{B}_{H;n} \) with explicit expression on localized observable $X_n = \sum_{i,j}X_{n;ij} e_ie_j^{\dagger}\in\mathcal{B}_{H;n}$, $Y_{n} = \sum_{i,j}\, X_{n;ij} |i\rangle\langle j| \in \mathcal{B}_{O; n}$   by:

$$
\mathcal{E}_{H, O;n}(X_n\otimes Y_{n}) = \sum_{i,j}\sum_{k,l}\overline{\chi^{[n]}_{i}(k)}\, \chi^{[n]}_{j}(l)\, X_{n;ij}Y_{n;kl} e_ie_j^{\dagger}
$$
 Its dual map, \( \mathcal{E}_{O;n}^{\dagger}  \equiv V^{[n]}_{O} (\cdot) V^{[n]\, \dagger}_{O}\), acts as a quantum channel from \( \mathcal{B}_{H;n}^{\dagger} \) to \( (\mathcal{B}_{H;n} \otimes \mathcal{B}_{O;n})^{\dagger} \).
Given an initial state \(\varphi_1\) on the first hidden algebra \(\mathcal{B}_{H; 1}\), such as \(\varphi_1(X) = \Tr(\rho_\pi X)\) with \(\rho_1 = \sum_i \pi_i e_i e_i^\dagger\), the triplet \((\varphi_1, (\mathcal{E}_{H;n})_n, (\mathcal{E}_{H,O;n})_n)\) defines an EQMM $\varphi_{H,O}$ within the infinite tensor product algebra
\(
\mathcal{B}_{H,O} = \bigotimes_{n \geq 1} (\mathcal{B}_{H,n} \otimes \mathcal{B}_{O,n})
\)
as described in \cite{AGLSQuan24, SS23}. Here, \(\mathcal{B}_{H,n}\) and \(\mathcal{B}_{O,n}\) are isomorphic copies of the hidden algebra \(\mathcal{B}_H\) and observation algebra \(\mathcal{B}_O\), respectively.
The information  of the correlations of the EHMM $\varphi_{H,O}$  is summarized by the partial isometries (\ref{VH}) and (\ref{VOn}) the following superposition states of the form
\begin{equation}\label{PsiHOn}
 \big| \Psi_{H,O;n}\big\rangle = \sum_{i_1,\dots, i_{n+1}}\sum_{k_1,\dots, k_n} \sqrt{\pi_{i_1}} U^{[1]}_{i_1i_2}\cdots U^{[n]}_{i_ni_{n+1}}\, \chi_{i_1}(k_1)\cdots \chi_{i_n}(k_n)e_{i_1\dots, i_{n+1}} \otimes |k_1\cdots k_n\rangle
\end{equation}
One can see that $\big| \Psi_{H,O;n}\big\rangle$ belongs to the tensor Hilbert $\mathcal{H}^{\otimes(n+1)}\otimes \mathcal{K}^{\otimes n}$.
\begin{theorem}
  In the above notations. The vector $\big| \Psi_{H,O;n}\big\rangle$ is a unit vector and for every local observable $Z\in\mathcal{B}_{H,O; loc}$ one has
  \begin{equation}\label{eq_phipsi}
    \varphi_{H,O}(Z) = \lim_{n\to\infty}\, \big\langle  \Psi_{H,O;n} \big|Z\otimes \id\big| \Psi_{H,O;n}\big\rangle
  \end{equation}
\end{theorem}
\begin{proof}
  See \cite{SS23}.
\end{proof}
For two Hilbert spaces \( \mathcal{U} \) and \( \mathcal{V} \), equipped with the respective inner products \( \langle \cdot, \cdot \rangle_\mathcal{U} \) and \( \langle \cdot, \cdot \rangle_\mathcal{V} \), the partial inner product map relative to a fixed unit vector \( u_0 \in \mathcal{U} \) is defined as follows:
\[
\langle \cdot \mid u_0 \rangle_{\mathcal{U}} : \mathcal{U} \otimes \mathcal{V} \to \mathcal{V}
\]
where, for any finite-rank tensor \( w \in \mathcal{U} \otimes \mathcal{V} \), written in the form \( w = \sum_{i=1}^n u_i \otimes v_i \) with \( u_i \in \mathcal{U} \) and \( v_i \in \mathcal{V} \), the map is given explicitly by:
\begin{equation} \label{Tu0}
\langle w \mid u_0 \rangle_{\mathcal{U}} = \sum_{i=1}^n \langle u_i, u_0 \rangle_\mathcal{U} \, v_i\in \mathcal{V}
\end{equation}

This operation, known as the partial inner product with respect to \( u_0 \), yields a vector in \( \mathcal{V} \). Within the quantum mechanical framework, it can be interpreted as a vector-valued partial measurement on the composite quantum system \( \mathcal{U} \otimes \mathcal{V} \), conditioned on the fixed state \( u_0 \) in the subsystem \( \mathcal{U} \).

For \( \mathcal{U} = \mathcal{H}^{\otimes (n+1)} \) and \( \mathcal{V} = \mathcal{K}^{\otimes (n+1)} \),
the observation process can be derived by partially measuring the EHMM \( \big|\Psi_{H,O;n}\big\rangle \) on the underlying entangled Markov process described by the state vector
\begin{equation}\label{PsiHn}
\big|\Psi_{H;n}\big\rangle  = \sum_{j_1,\dots, j_{n+1}}  \sqrt{\pi_{j_1}} U^{[1]}_{j_1j_2}\cdots U^{[n]}_{j_nj_{n+1}}e_{j_1\dots, j_{n+1}}
\end{equation}

In fact, one has
\begin{align*}
   \big\langle\Psi_{H,O;n} |\Psi_{H;n}\big\rangle_{\mathcal{H}^{\otimes(n+1)}}  &  = \sum_{i_1,\dots, i_{n+1}}\sum_{k_1,\dots, k_n} \pi_{i_1} |U^{[1]}_{i_1i_2}|^2\cdots |U^{[n]}_{i_ni_{n+1}}|^2\, \chi_{i_1}(k_1)\cdots \chi_{i_n}(k_n)  |k_1\cdots k_n\rangle \\
   &  = \sum_{i_1,\dots, i_{n}}\sum_{k_1,\dots, k_n} \pi_{i_1} \Pi^{[1]}_{i_1i_2} \cdots \Pi^{[n]}_{i_{n-1}i_{n}}\, \chi_{i_1}(k_1)\cdots \chi_{i_n}(k_n)  |k_1\cdots k_n\rangle
\end{align*}

This leads to the observation process, which is described by the following:
\begin{equation}\label{PsiOn}
\big| \Psi_{O;n}\big\rangle = \sum_{i_1,\dots, i_{n}}\sum_{k_1,\dots, k_n} \pi_{i_1} \Pi^{[1]}_{i_1i_2} \cdots \Pi^{[n]}_{i_{n-1}i_{n}}\, \chi_{i_1}(k_1)\cdots \chi_{i_n}(k_n)  |k_1\cdots k_n\rangle \in \mathcal{K}^{\otimes n}
\end{equation}

Conversely, by measuring the EHMM \( \big| \Psi_{H,O;n}\big\rangle \) given the observation process (\ref{PsiOn}), one gets
\begin{align*}
  \big\langle\Psi_{H,O;n} |\Psi_{O;n}\big\rangle_{\mathcal{K}^{\otimes(n)}} &= \sum_{\substack{i_1,\dots,i_{n+1} \\ j_1,\dots,j_{n+1}}}\sum_{k_1,\dots,k_n} \sqrt{\pi_{i_1}} U^{[1]}_{i_1i_2}\cdots U^{[n]}_{i_ni_{n+1}} \pi_{j_1} \Pi^{[1]}_{j_1j_2} \cdots \Pi^{[n]}_{j_{n-1}j_{n}}\\
& \times \chi_{j_1}(k_1)\cdots \chi_{j_n}(k_n)\chi_{i_1}(k_1)\cdots \chi_{i_n}(k_n)e_{i_1\dots, i_{n+1}}  \in \mathcal{H}^{\otimes(n+1)}
\end{align*}

One can observe that the expression of \( \big\langle\Psi_{H,O;n} |\Psi_{O;n}\big\rangle_{\mathcal{K}^{\otimes(n)}} \) is significantly different from the expression of the hidden Markov chain \( \big|\Psi_{H;n}\big\rangle \) given in (\ref{PsiHn}). This observation provides an intuitive insight into the entanglement of the state \( \big|\Psi_{H,O;n}\big\rangle \), as it is clearly not a product state. However, a precise measure of entanglement should be calculated to quantify its degree of entanglement, using measures such as the entanglement entropy and the negativity. These notions will be addressed in future work.

To ensure consistency with the notation used for MPS, we have shifted the initial system, often denoted as \(\mathcal{B}_{H,0}\) in the standard framework of QMCs and HQMMs, to \(\mathcal{B}_{H,1}\). This adjustment aligns the presentation with the conventions of MPS, providing a coherent and unified notation.

\section{Matrix Product States as  Observations of EHMMs}\label{sect-MPS-HMM}
In this section, we examine an  EHMMs characterized by the sequence of state vectors \(\big|\Psi_{H,O;n}\big\rangle\), as defined in (\ref{PsiHOn}). Alternatively, this model can be fully specified by the triplet \((\pi_1, (V_{H}^{[n]})_n, (V_{O}^{[n]})_n)\), where:
\begin{itemize}
    \item \(\pi_1 = \sum_{i} p_i e_i \) represents the initial distribution,
    \item \(V_{H}^{[n]}\) denotes the hidden partial isometries as given in (\ref{VH}), and
    \item \(V_{O}^{[n]}\) represents the observation partial isometries as defined in (\ref{VOn}).
\end{itemize}
We demonstrate that an important MPSs can be derived by applying a partial inner product operator to the sequence of state vectors \(\big|\Psi_{H,O;n}\big\rangle\).
Mainly we will treat the case of MPSs represented by  a quantum state \(|\psi_N\rangle\) in the observation Hilbert space \(\mathcal{K}^{\otimes N}\), where \(N\) is the number of subsystems. These states are expressed d as:
\begin{equation}\label{eq_psiN}
|\psi_N\rangle = \sum_{k_1, k_2, \ldots, k_N} \mathrm{Tr}(A_{k_1}^{[1]} A_{k_2}^{[2]} \cdots A_{k_N}^{[N]}) |k_1 k_2 \cdots k_N\rangle,
\end{equation}
where \(A_k^{[n]} = (a^{[n]}_{k;ij})\) are \(m \times m\) matrices acting on the observation  Hilbert space \(\mathcal{K}\) (so-called auxiliary in the language of MPS), \(k \in I_O\) indexes the local basis states \(\{|k\rangle\}\) and $n\ge 0$. For consistency of the states $|\psi_N\rangle$, the tensors \(A_k^{[n]}\) satisfy the following gauge condition:
\begin{equation}\label{eq_sumAkAkd1}
\sum_{k \in I_O} A^{[n]}_k A^{[n]\,\dagger}_k = \id_m
\end{equation}
 This approach establishes a direct connection between the EHMM framework and the MPS representation, highlighting the interplay between partial measurements and entanglement structure.

\begin{theorem}\label{thm_main1}
Let an EHMM be defined by the triplet \( (\pi_1, (V_{H}^{[n]})_n, (V_{O}^{[n]})_n) \), where \( \pi_1 = \sum_{i} \pi_i e_i   \) is an initial state, \( V_{H}^{[n]}: e_i \mapsto \sum_{j} U^{[n]}_{ij} e_i \otimes e_j \) are unistochastic hidden partial isometries, and \( V_{O}^{[n]}: e_i \mapsto \sum_{k} \chi^{[n]}_{i}(k) e_i \otimes |k\rangle \) are observation partial isometries.

For any \( N \in \mathbb{N} \), the MPS generated by matrices \( A_k^{[n]} = (a^{[n]}_{k; ij}) \) with entries
\begin{equation}\label{eq_aij}
a^{[n]}_{k; ij} = U^{[n]}_{ij} \chi^{[n]}_{i}(k)
\end{equation}
can be expressed as a partial measurement of the EHMM state \( |\Psi_{H,O;n}\rangle \) with respect to the vector \( E_{N,n} \in \mathcal{H}^{\otimes(N+1)} \otimes \mathcal{K}^{\otimes(n-N)} \):
\begin{equation}\label{eq_ENn}
E_{N,n} = \sum_{i_1,\dots, i_{n+1}} \sum_{k_{N+1}, \dots, k_n} \frac{1}{\sqrt{\pi_{i_1}}} \bigg( \prod_{\ell=N+1}^{n} U_{i_{\ell}i_{\ell+1}} \chi_{i_{\ell}}(k_\ell) \bigg) \delta_{i_{N+1}, i_1} e_{i_1,\dots,i_{n+1}} \otimes |k_{N+1}, \dots, k_n\rangle
\end{equation}
such that:
\begin{equation}\label{eq_PsiHOnENn}
\big\langle \Psi_{H,O;n} \mid E_{N,n} \big\rangle_{\mathcal{H}^{\otimes (n+1)} \otimes \mathcal{K}^{\otimes (n-N)}} = \sum_{k_1, \dots, k_N} \mathrm{Tr}(A_{k_1}^{[1]} \cdots A_{k_N}^{[N]}) \, |k_1 k_2 \cdots k_N\rangle, \qquad \forall n \geq N
\end{equation}
Additionally, the gauge condition (\ref{eq_sumAkAkd1}) is satisfied.
\end{theorem}
\begin{proof}
Let \( U^{[n]} = (U^{[n]}_{ij}) \) be a unitary matrix and \( \chi^{[n]}_{i}(k) \) a complex-valued function. defining the hidden and observation partial isometries \( V_H^{[n]} \) and \( V_O^{[n]} \), respectively.
The associated stochastic matrices are defined by
\[
\Pi^{[n]}_{ij} = \big|U^{[n]}_{ij}\big|^2, \quad Q^{[n]}_{i}(k) = \big|\chi^{[n]}_{i}(k)\big|^2
\]
where \( \Pi^{[n]} \) is the transition matrix for the hidden process, and \( Q^{[n]} \) represents the emission probabilities. Both matrices are stochastic by construction.

We now define the entries of the MPS matrices \( A_k^{[n]} = (a^{[n]}_{k; ij}) \) by:
\begin{equation}\label{eq_aij_proof}
a^{[n]}_{k; ij} = U^{[n]}_{ij} \chi^{[n]}_{i}(k)
\end{equation}

Using the definition of \( A_k^{[n]} \), we compute:
\[
\sum_k A_k^{[n]} A_k^{[n]\, \dagger} = \sum_k \sum_{i',j} U^{[n]}_{ii'} \chi^{[n]}_{i}(k) \overline{U^{[n]}_{ji'}} \overline{\chi^{[n]}_{j}(k)} \, e_i e_j^\dagger
\]
Since \( U^{[n]} \) is unitary, we have \( \sum_{i'} U^{[n]}_{ii'} \overline{U^{[n]}_{ji'}} = \delta_{ij} \). Substituting this, we obtain:
\[
\sum_k A_k^{[n]} A_k^{[n]\, \dagger} = \sum_i \sum_k Q^{[n]}_{i}(k) \, e_i e_i^\dagger = \id_m
\]
which confirms the gauge condition (\ref{eq_sumAkAkd1}). The EHMM state vector is given by:
\[
\big|\Psi_{H,O;n}\big\rangle = \sum_{i_1,\dots,i_{n+1}} \sum_{k_1,\dots,k_n} \sqrt{\pi_{i_1}} \prod_{\ell=1}^n U^{[\ell]}_{i_\ell i_{\ell+1}} \chi^{[\ell]}_{i_\ell}(k_\ell) \, e_{i_1,\dots,i_{n+1}} \otimes |k_1 \cdots k_n\rangle
\]
Splitting the product into two parts, for \( 1 \leq \ell \leq N \) and \( N+1 \leq \ell \leq n \), we write:
\[
\big|\Psi_{H,O;n}\big\rangle = \sum_{i_1,\dots,i_N} \sum_{k_{N+1},\dots,k_n} \sqrt{\pi_{i_1}} \prod_{\ell=1}^N a^{[\ell]}_{k_\ell; i_\ell i_{\ell+1}} \prod_{\ell=N+1}^n a^{[\ell]}_{k_\ell; i_\ell i_{\ell+1}} \, e_{i_1,\dots,i_{n+1}} \otimes |k_1 \cdots k_n\rangle
\]

The vector \( E_{N,n} \) is defined as:
\[
E_{N,n} = \sum_{i_1,\dots,i_{n+1}} \sum_{k_{N+1},\dots,k_n} \frac{1}{\sqrt{\pi_{i_1}}} \prod_{\ell=N+1}^n a^{[\ell]}_{k_\ell; i_\ell i_{\ell+1}} \delta_{i_{N+1},i_1} \, e_{i_1,\dots,i_{n+1}} \otimes |k_{N+1} \cdots k_n\rangle
\]

Projecting \( |\Psi_{H,O;n}\rangle \) onto \( \mathcal{K}^{\otimes N} \) via \( E_{N,n} \), we compute:
\[
\big\langle \Psi_{H,O;n} \mid E_{N,n} \big\rangle_{\mathcal{H}^{\otimes(n+1)} \otimes \mathcal{K}^{\otimes(n-N)}} = \sum_{i_1,\dots,i_N} \sum_{k_1,\dots,k_N} \prod_{\ell=1}^N a^{[\ell]}_{k_\ell; i_\ell i_{\ell+1}} \delta_{i_{N+1},i_1} \prod_{\ell=N+1}^n |a^{[\ell]}_{k_\ell; i_\ell i_{\ell+1}}|^2 \, |k_1 \cdots k_N\rangle
\]

Using \( \big|a^{[\ell]}_{k_\ell; i_\ell i_{\ell+1}}\big|^2 = \Pi^{[\ell]}_{i_\ell i_{\ell+1}} Q^{[\ell]}_{i_\ell}(k_\ell) \) and the stochasticity of \( \Pi^{[\ell]} \) and \( Q^{[\ell]} \), it follows that:
\[
\sum_{i_{N+2},\dots,i_{n+1}} \sum_{k_{N+1},\dots,k_n} \prod_{\ell=N+1}^n \big|a^{[\ell]}_{k_\ell; i_\ell i_{\ell+1}}\big|^2 = 1
\]

Thus:
\begin{eqnarray*}
\big\langle \Psi_{H,O;n} \mid E_{N,n} \big\rangle_{\mathcal{H}^{\otimes (n+1)}\otimes\mathcal{K}^{\otimes (n-N)}} &=& \sum_{k_1,\dots, k_N} \sum_{i_1,i_{N+1}}\bigg[\prod_{\ell=1}^{N}A^{[\ell]}_{k}\bigg]_{i_1i_{N+1}}\delta_{i_1,i_{N+1}}   |k_1\cdots k_N\rangle\\
&=& \sum_{k_1,\dots, k_N} \Tr\bigg(\prod_{\ell=1}^{N}A^{[\ell]}_{k}\bigg) \, |k_1\cdots k_N\rangle\\
\end{eqnarray*}

This establishes (\ref{eq_PsiHOnENn}) and completes the proof.
\end{proof}
\begin{remark}
  The connection between MPS and hidden quantum Markov models is essential for understanding the structure and extension of MPS within a more general algebraic framework. In particular, the construction of MPS can be naturally incorporated into the formalism of QMCs.

  In the setting of the above theorem, when the number of sites satisfies \( N = n \), we consider the vector
  \begin{equation}
    E_{N} = \sum_{i_1,\cdots, i_{N+1}}\frac{1}{\sqrt{\pi_{i_1}}} e_{i_1, \ldots, i_{N+1}}.
  \end{equation}
  This vector plays a crucial role in the characterization of MPS within the framework of hidden quantum Markov models. Specifically, it satisfies the inner product relation:
  \begin{equation}
    \big\langle \Psi_{H,O;N} \mid E_{N} \big\rangle_{\mathcal{H}^{\otimes (n+1)}}  = \sum_{k_1, \dots, k_N} \mathrm{Tr}(A_{k_1}^{[1]} \cdots A_{k_N}^{[N]}) \, |k_1 k_2 \cdots k_N\rangle.
  \end{equation}
  This result highlights how MPS emerge naturally within the formalism of hidden quantum Markov models, providing a deeper insight into their structure. Furthermore, this perspective enables a systematic extension of MPS to the full observation algebra \( \mathcal{B}_{O;\mathbb{N}} \), offering a mathematically rigorous framework for investigating the  thermodynamic limit  in quantum spin systems.
\end{remark}

\subsection{Distinguishability  of $|\psi_N\rangle$ and $|\Psi_{O;N}\rangle$}\label{sect-Dist}
This section is devoted to study the difference between the structures of the two states  $|\psi_N\rangle$ and $|\Psi_{H,O;N}\rangle$ since both belong to the observation  space $\mathcal{H}^{\otimes(N)}$.
The quantum relative entropy \cite{OP} between two density operators $\rho$ and $\sigma$ on the same Hilbert space is defined as:
\[
S(\rho \| \sigma) =
\begin{cases}
\mathrm{Tr}(\rho \log \rho) - \mathrm{Tr}(\rho \log \sigma), & \text{if } \mathrm{supp}(\rho) \subseteq \mathrm{supp}(\sigma), \\
+\infty, & \text{otherwise}.
\end{cases}
\]
where $\rho$ is the state of interest and $\sigma$ is the reference state. Quantum relative entropy is non-negative ($S(\rho \| \sigma) \geq 0$) and equals zero if and only if $\rho = \sigma$.

\begin{lemma}\label{lem-data}
Let \(\rho, \sigma \in \mathcal{B}(\mathcal{H})\) be density matrices and \(\mathcal{E}: \mathcal{B}(\mathcal{H}) \to \mathcal{B}(\mathcal{H}')\) a quantum channel. The quantum relative entropy satisfies:
\begin{equation}\label{eq_data}
 S(\mathcal{E}(\rho) \| \mathcal{E}(\sigma)) \leq S(\rho \| \sigma)
\end{equation}
The inequality above is commonly referred to as the Data Processing Inequality.
\end{lemma}
\begin{proof}[Proof Sketch]
Using the Stinespring dilation theorem, the CPTP map $\Phi$ can be written as $\Phi(\rho) = \mathrm{Tr}_E[U (\rho \otimes |0\rangle\langle 0|) U^\dagger]$. The monotonicity of relative entropy under partial trace ensures that
\[
S(\Phi(\rho) \| \Phi(\sigma)) \leq S(U (\rho \otimes |0\rangle\langle 0|) U^\dagger \| U (\sigma \otimes |0\rangle\langle 0|) U^\dagger).
\]
Finally, the invariance of relative entropy under unitary transformations gives $S(U (\rho \otimes |0\rangle\langle 0|) U^\dagger \| U (\sigma \otimes |0\rangle\langle 0|) U^\dagger) = S(\rho \| \sigma)$, proving the inequality. For further details, see \cite{N10}.
\end{proof}

The density matrix of the \((H,O)\)-density operator \(\rho_{H,O;N} = |\Psi_{H,O;N}\rangle\langle\Psi_{H,O;N}|\) for the EHMM \(|\Psi_{H,O;N}\rangle\) (given by (\ref{PsiHOn})) can be analyzed as follows:

\begin{lemma}\label{lem_rhoON}
Let \(\rho_{O;N} = \Tr_{\mathcal{H}^{\otimes(N+1)}}(\rho_{H,O;N})\) denote the partial observation density matrix obtained by tracing out the Hilbert space \(\mathcal{H}^{\otimes(N+1)}\) from the joint density matrix \(\rho_{H,O;N}\). Then, the partial observation density matrix is given by:
\begin{equation}\label{rho_ON}
\rho_{O;N} = \sqrt{m}\sum_{k_1, \dots, k_N} \pi^\dagger \big( A^{[1]}_{k_1} \diamond \overline{A^{[1]}_{k'_1}} \big) \cdots \big( A^{[N]}_{k_N} \diamond \overline{A^{[N]}_{k'_N}} \big) e   \, |k_1 \cdots k_N\rangle \langle k'_1 \cdots k'_N|
\end{equation}
where \(A_{k_\ell}^{[\ell]}\) is an \(m \times m\) matrix with entries \(a^{[\ell]}_{k; ij} = U_{ij}^{[\ell]} \, \chi^{[\ell]}_i(k)\) and $\pi^\dagger = \sum_{i}\pi_i e_i^\dagger$ and $e = \frac{1}{\sqrt{m}}\sum_{j}e_j$. Here \( \diamond \) denotes the Schur product between the matrices
\end{lemma}

\begin{proof}
The joint \((H, O)\)-density matrix \(\rho_{H,O;N}\) is expressed as:
\[
\rho_{H,O;N} = \sum_{\substack{i_1, \dots, i_{N+1} \\ i'_1, \dots, i'_{N+1}}} \sum_{\substack{k_1, \dots, k_N \\ k'_1, \dots, k'_N}}\sqrt{\pi_{i_1}}\sqrt{\pi_{i'_1}}
\Bigg( \prod_{\ell=1}^N U^{[\ell]}_{i_\ell i_{\ell+1}} \chi^{[\ell]}_{i_\ell}(k_\ell)
\prod_{\ell=1}^N \overline{U^{[\ell]}_{i'_\ell i'_{\ell+1}} \chi^{[\ell]}_{i'_\ell}(k'_\ell)} \Bigg)
\]
\[\times
e_{i_1 \dots i_{N+1}} e^\dagger_{i'_1 \dots i'_{N+1}} \otimes |k_1 \cdots k_N\rangle \langle k'_1 \cdots k'_N|
\]

Tracing out the Hilbert space \(\mathcal{H}^{\otimes(N+1)}\), the partial density matrix \(\rho_{O;N}\) is obtained as:
\[
\rho_{O;N} = \Tr_{\mathcal{H}^{\otimes(N+1)}}(\rho_{H,O;N})
\]

Substituting \(\rho_{H,O;N}\), we get:
\[
\rho_{O;N} = \sum_{i_1, \dots, i_{N+1} } \sum_{\substack{k_1, \dots, k_N \\ k'_1, \dots, k'_N}}\pi_{i_1}
\prod_{\ell=1}^N\Big( U^{[\ell]}_{i_\ell i_{\ell+1}} \chi^{[\ell]}_{i_\ell}(k_\ell)
 \overline{U^{[\ell]}_{i_\ell i_{\ell+1}} \chi^{[\ell]}_{i_\ell}(k'_\ell)} \Big)
 |k_1 \cdots k_N\rangle \langle k'_1 \cdots k'_N|
\]

Rewriting in terms of the matrices $A^{[\ell]}_{k_{\ell}}$ with  elements \(a^{[\ell]}_{k_{\ell}; i_\ell i_{\ell+1}} = U^{[\ell]}_{i_\ell i_{\ell+1}} \chi^{[\ell]}_{i_\ell}(k_\ell)\), the expression becomes:
\[
\rho_{O;N} = \sum_{\substack{k_1, \dots, k_N \\ k'_1, \dots, k'_N}} \sum_{i_1, \dots, i_{N+1}}
\pi_{i_1}\prod_{\ell=1}^N\big( a^{[\ell]}_{k_\ell; i_\ell i_{\ell+1}}
\overline{a^{[\ell]}_{k'_\ell; i_\ell i_{\ell+1}}}\big) \, |k_1 \cdots k_N\rangle \langle k'_1 \cdots k'_N|
\]

Using the notation for the Schur product notation, we can express this as:
\[
\rho_{O;N} = \sum_{\substack{k_1, \dots, k_N \\ k'_1, \dots, k'_N}} \sum_{i_1, \dots, i_{N+1}}\pi_{i_1}
\prod_{\ell=1}^N \big( A^{[\ell]}_{k_\ell} \diamond \overline{A^{[\ell]}_{k'_\ell}} \big)_{i_\ell i_{\ell+1}}
|k_1 \cdots k_N\rangle \langle k'_1 \cdots k'_N|
\]

Finally, summing over the indices \(i_1, \dots, i_{N+1}\) and writing the result in terms of the \(e_i\)-basis, we recover the expression in  \eqref{rho_ON}, concluding the proof.
\end{proof}

\begin{theorem}\label{thm_entropyInequality}
Under the notations and assumptions of Theorem \ref{thm_main2}, let \( \rho_{N} = \frac{1}{m} |\psi_N\rangle \langle \psi_N| \) represent the density operators of the MPS and \( \rho_{O;N} \) denote the observation density matrix. The relative entropy between these density operators satisfies the following inequality:
\begin{equation}\label{eq_entrIneq}
S(\rho_N \| \rho_{O;N}) \geq \frac{1}{m} \sum_{k_1, \dots, k_N} \left| \mathrm{Tr}\left( A_{k_1}^{[1]} \cdots A_{k_N}^{[N]} \right) \right|^2 \log \left( \frac{\left| \mathrm{Tr}\left( A_{k_1}^{[1]} \cdots A_{k_N}^{[N]} \right) \right|^2}{m^{3/2} \pi^\dagger \left( \prod_{\ell=1}^N \left| A_{k_\ell}^{[\ell]} \right|^2_{\diamond} \right) e} \right)
\end{equation}
Here, \( \left| A_{k_\ell}^{[\ell]} \right|^2_{\diamond} = A_{k_\ell}^{[\ell]} \diamond \overline{A_{k_\ell}^{[\ell]}} \),  \( \pi^\dagger \) and \( e \) as defined in Lemma \ref{lem_rhoON}, and the summation spans all possible sequences \( k_1, \dots, k_N \).
\end{theorem}
\begin{proof}
The density  matrix associated with the MPS $|\psi_N\rangle$  given by (\ref{eq_psiN}) is given by  $\rho_{N} = \frac{1}{m} |\psi_N\rangle\langle\psi_N |$   then $\rho_{N}$ becomes
$$
\rho_{N} =  \sum_{\substack{k_1,\dots, k_{N} \\ k'_1,\dots, k'_{N}}} \mathrm{Tr}\big(A_{k_1}^{[1]} A_{k_2}^{[2]} \cdots A_{k_N}^{[N]}\big)\overline{\mathrm{Tr}\big(A_{k'_1}^{[1]} A_{k'_2}^{[2]} \cdots A_{k'_N}^{[N]}\big)}\, |k_1  \cdots k_N\rangle\langle k'_1  \cdots k'_N|
$$
Taking into account that the expressions of the tensors $A^{[n]}_{k_n}$ satisfy (\ref{eq_aij}), one gets
\[
\rho_{N} = \sum_{\substack{k_1,\dots, k_{N} \\ k'_1,\dots, k'_{n}}}\sum_{\substack{i_1,\dots, i_{N} \\ i'_1,\dots, i'_{n}}}\big(U^{[1]}_{i_1i_2}\chi_{i_1}(k_1)\cdots U^{[N]}_{i_Ni_{1}}\chi^{[\ell]}_{i_N}(k_N)\big)\ \big(\overline{U^{[1]}_{i'_1i'_2}\chi^{[\ell]}_{i'_1}(k'_1)}\cdots \overline{U^{[N]}_{i'_Ni'_{1}}\chi^{[N]}_{i'_N}(k'_N)}\big)
\]
\[
\times \big|k_1\cdots k_N\big\rangle\big\langle k'_1\cdots k'_N\big|
\]
The diagonal observation density matrix is defined
$$
\rho_{O;N, diag} =  \sum_{k_1, \dots, k_N} \sum_{i,j} e_i^\dagger \big( A^{[1]}_{k_1} \diamond \overline{A^{[1]}_{k_1}} \big) \cdots \big( A^{[N]}_{k_N} \diamond \overline{A^{[N]}_{k_N}} \big) e_j \, |k_1 \cdots k_N\rangle \langle k_1 \cdots k_N|
$$
then its logarithm is
$$
\log(\rho_{O;N, diag}) =  \sum_{k_1,\dots, k_N}\log\Big(\sum_{i,j} e_i^\dagger \big( A^{[1]}_{k_1} \diamond \overline{A^{[1]}_{k_1}} \big) \cdots \big( A^{[N]}_{k_N} \diamond \overline{A^{[N]}_{k_N}} \big) e_j\Big)\, |k_1 \cdots k_N\rangle\langle k_1 \cdots k_N|
$$
and
$$
\rho_{N; diag} =  \frac{1}{m}\sum_{ k_1,\dots, k_{N}  }\big|\mathrm{Tr}\big(A_{k_1}^{[1]}  \cdots A_{k_N}^{[N]}\big)\big|^2 \, |k_1  \cdots k_N\rangle\langle k_1  \cdots k_N|
$$

It follows that
$$
S(\rho_{N, diag} \| \rho_{O;N, diag}) = \Tr(\rho_{N, diag} \log(\rho_{N, diag})\big) - \Tr(\rho_{N, diag} \log(\rho_{O;N, diag})\big)
$$
$$
= \frac{1}{m}\sum_{ k_1,\dots, k_{N}  }\big|\mathrm{Tr}\big(A_{k_1}^{[1]} A_{k_2}^{[2]} \cdots A_{k_N}^{[N]}\big)\big|^2\log\left(\frac{\big|\mathrm{Tr}\big(A_{k_1}^{[1]} A_{k_2}^{[2]} \cdots A_{k_N}^{[N]}\big)\big|^2}{m^{3/2}\,\pi^\dagger \big( A^{[1]}_{k_1} \diamond \overline{A^{[1]}_{k_1}} \big) \cdots \big( A^{[N]}_{k_N} \diamond \overline{A^{[N]}_{k_N}} \big) e}\right)
$$

The diagonal projection $\Phi$  from $\mathcal{B}(\mathcal{K}^{\otimes N})$  onto the diagonal algebra $\mathcal{D}_f(\mathcal{K}^{\otimes N})$, spanned by the elements $|k_1\cdots k_N\rangle\langle k_1\cdots k_N| $ and  defined by
$$
\Phi(Y) = \sum_{k_1,\dots, k_N}|k_1\cdots k_N\rangle\langle k_1\cdots k_N| \,Y \, |k_1\cdots k_N\rangle\langle k_1\cdots k_N|
$$
is  a clearly a quantum channel, even  bistochastic channel and $\rho_{O;N, diag}= \Phi(\rho_{O;N})$ and $\rho_{N, diag}= \Phi(\rho_{N})$. Considering the data processing inequality (\ref{eq_data}). This leads to (\ref{eq_entrIneq}) and finishes the proof.
\end{proof}

\section{EHMMs Arising from Matrix Product States}\label{sect-EMM-MPS}
This section explores the converse direction of the previous theorem. Starting with a MPS  satisfying the gauge condition, we construct a relevant  EHMM.

\begin{theorem}\label{thm_main2}
Let \( \{A^{[n]}_{k}\}_{1 \leq k \leq d, n \in \mathbb{N}} \) be a sequence of \( m \times m \) matrices such that for each \( n \), the following gauge condition holds:
\begin{equation}\label{eq_SAkA1}
\sum_{k} A^{[n]}_{k} A^{[n]\dagger}_{k} = \id_m
\end{equation}
Then, the maps \( V^{[n]}_{H} : \mathcal{H} \to \mathcal{H} \otimes \mathcal{H} \) and \( V^{[n]}_{H,O} : \mathcal{H} \to \mathcal{H} \otimes \mathcal{K} \) defined by:
\begin{equation}\label{eq_VHVOn}
V^{[n]}_{H}(e_{i}) = \sum_{j} \sqrt{\sum_{k} |a_{k;ij}^{[n]}|^2} \, e_{i} \otimes e_{j}, \quad
V^{[n]}_{H,O}(e_{i}) = \sum_{k} \sqrt{\sum_{j} |a_{k;ij}^{[n]}|^2} \, e_{i} \otimes |k\rangle
\end{equation}
are partial isometries and define an EHMM on the algebra \( \mathcal{B}_{H,O} \). Here, \( \sqrt{\cdot} \) denotes any complex square root, not necessarily the positive one.
\end{theorem}
\begin{proof}
Let the entries of the matrix \( A_{k}^{[n]} \) be denoted by \( a_{k;ij}^{[n]} \). Define the following quantities:
\begin{equation}\label{eq_PiQd}
\Pi'^{[n]}_{ij} := \sum_{k} |a_{k;ij}^{[n]}|^2, \quad Q'^{[n]}_i(k) := \sum_{j} |a_{k;ij}^{[n]}|^2
\end{equation}
The gauge condition (\ref{eq_SAkA1}) implies that:
\begin{equation*}
\sum_{k} \sum_{j} a_{k;ij}^{[n]} \overline{a^{[n]}_{k;i'j}} = \delta_{i,i'}
\end{equation*}
In particular, for \( i = i' \), we obtain:
\begin{equation*}
\sum_{k} \sum_{j} |a_{k;ij}^{[n]}|^2 = 1
\end{equation*}
This shows that the matrix \( \Pi'^{[n]} = [\Pi'^{[n]}_{ij}] \) is an \( m \times m \) stochastic matrix, and \( Q'^{[n]} = [Q'^{[n]}_i(k)] \) is an \( m \times d \) stochastic matrix.

Now consider the EHMM associated with the pair \( \big((\Pi'^{[n]})_n, (Q'^{[n]})_n\big) \).

The  emission transition expectation  is defined as a map from \( \mathcal{B}_{H;n} \otimes \mathcal{B}_{O;n} \) to \( \mathcal{B}_{H;n} \):
\begin{align*}
\mathcal{E}_{H,O;n}(X_n \otimes Y_n) &= V^{[n]}_{H,O;n} (X_n \otimes Y_n) V^{[n]\dagger}_{H,O;n} \\
&= \sum_{(i_n,j_n) \in I_H^2} \sum_{k_n, l_n \in I_O} \overline{\sqrt{\sum_{j} |a^{[n]}_{k_n;i_nj}|^2}} \sqrt{\sum_{j} |a^{[n]}_{l_n;j_nj}|^2} \, x_{n;i_nj_n} y_{n;k_nl_n} \, |i_n\rangle \langle j_n| \\
&= \sum_{(i_n,j_n) \in I_H^2} \sum_{k_n, l_n \in I_O} \overline{\sqrt{Q'^{[n]}_{i_n}(k_n)}} \sqrt{Q'^{[n]}_{j_n}(l_n)} \, x_{n;i_nj_n} y_{n;k_nl_n} \, |i_n\rangle \langle j_n|
\end{align*}

The  hidden transition expectation  is defined as a map from \( \mathcal{B}_{H;n} \otimes \mathcal{B}_{H;n+1} \) to \( \mathcal{B}_{H;n} \):
\begin{align*}
\mathcal{E}_{H;n}(X_n \otimes X_{n+1}) &= V^{[n]}_{H;n} (X_n \otimes X_{n+1}) V^{[n]\dagger}_{H;n} \\
&= \sum_{\substack{i_n,j_n \\ i_{n+1},j_{n+1}}} \overline{\sqrt{\sum_{k} |a^{[n]}_{k;i_ni_{n+1}}|^2}} \sqrt{\sum_{k} |a^{[n]}_{k;j_nj_{n+1}}|^2} \, x_{n;i_nj_n} x_{n+1;i_{n+1}j_{n+1}} \, |i_n\rangle \langle j_n| \\
&= \sum_{\substack{i_n,j_n \\ i_{n+1},j_{n+1}}} \overline{\sqrt{\Pi'^{[n]}_{i_n,i_{n+1}}}} \sqrt{\Pi'^{[n]}_{j_n,j_{n+1}}} \, x_{n;i_nj_n} x_{n+1;i_{n+1}j_{n+1}} \, |i_n\rangle \langle j_n|
\end{align*}

Finally, by specifying an initial state \( \rho \), the quadruplet \( (\rho, (\mathcal{E}_{H;n})_n, (\mathcal{E}_{H,O;n})_n) \) defines an EHMM. This completes the proof.
\end{proof}

\begin{example}
Consider
$$
A_1^{[n]} = \left[\begin{array}{cc}
                    \cos(\theta_n) & 0 \\
                    0 & 1
                  \end{array}\right]\quad; \quad A_2^{[n]} = \left[\begin{array}{cc}
                    0 & \sin(\theta_n) \\
                    0 & 0
                  \end{array}\right]
$$
It is clear that $A_1^{[n]}$ and $A_2^{[n]}$ satisfy (\ref{eq_SAkA1}). One has $d=m=2$, consider $\mathcal{H}=\mathcal{K} = \mathbb{C}^2$ with orthonormal basis $e_1 = |1\rangle =  \left[\begin{array}{cc}
                    1 \\
                    0
                  \end{array}\right]\quad; \quad e_2 = |2\rangle =  \left[\begin{array}{cc}
                    0  \\
                    1
                  \end{array}\right]$. \\
From (\ref{eq_PiQ}) one finds
$$
\Pi^{[n]} = \left[\begin{array}{cc}
                    \cos^2(\theta_n) & \sin^2(\theta_n) \\
                    0 & 1
                  \end{array}\right] \quad ; \quad Q^{[n]} = \left[\begin{array}{cc}
                    \cos^2(\theta_n) & \sin^2(\theta_n) \\
                    1 & 0
                  \end{array}\right]
$$
The associated hidden partial isometry is given by
\[
V_{H}^{[n]}(e_1) =  |\cos(\theta_n)|e_1\otimes e_1 + |\sin(\theta_n)|e_1\otimes e_2 \quad; \quad  V_{H}^{[n]}(e_2) = e_2
\]
and the associated observation partial isometry is given by
\[
V_{O}^{[n]}(e_1) =  |\cos(\theta_n)|e_1\otimes |1\rangle + |\sin(\theta_n)|e_1\otimes |2\rangle \quad; \quad  V_{O}^{[n]}(e_2) = e_1
\]
\end{example}
The HMM \eqref{eq_PiQd} emerging from the MPS in Theorem \ref{thm_main2} coincides with the one in Theorem \ref{thm_main1} whenever the tensors satisfy the decomposition \eqref{eq_aij}. As we will see in the next section for GHZ and cluster states, this condition holds, ensuring a consistent structure.
However, in the case of AKLT, the decomposition \eqref{eq_aij} is not applicable, leading to a different HMM. This distinction reflects deeper structural differences in the associated processes and highlights new aspects of correlations and entanglement beyond standard representations.

\section{Application to Notable Matrix Product States}

In this section, we apply the previously established results to prominent examples of MPS, illustrating their relevance and implications. By analyzing well-known cases from the literature, we highlight how our framework naturally encompasses these states and reveals new structural insights.

\subsection{GHZ State Representation}

The Greenberger-Horne-Zeilinger (GHZ) state for an $N$-qubit system is formally expressed as:
\[
\big|{\text{GHZ}}\big\rangle_N = \frac{1}{\sqrt{2}} \left( \big|{0}\big\rangle^{\otimes N} +  \big|{1}\big\rangle^{\otimes N} \right).
\]

To formulate this state as an MPS under periodic boundary conditions (PBC), we define the site-dependent matrices $A^{[n]}_{k}$ as follows:
\[
A^{[1]}_{0} = \frac{1}{\sqrt{2}}
\begin{pmatrix}
1 & 0 \\
0 & 0
\end{pmatrix}, \quad
A^{[1]}_{1} = \frac{1}{\sqrt{2}}
\begin{pmatrix}
0 & 0 \\
0 & 1
\end{pmatrix}
\]
\[
A^{[n]}_{0} =
\begin{pmatrix}
1 & 0 \\
0 & 0
\end{pmatrix}, \quad
A^{[n]}_{1} =
\begin{pmatrix}
0 & 0 \\
0 & 1
\end{pmatrix}, \quad \forall n\geq 2.
\]

The resulting MPS form of the GHZ state can be expressed as:
\[
\big|{\text{GHZ}}\big\rangle_{N} = \frac{1}{\sqrt{2}}\sum_{k_1, \dots, k_N}\Tr\left( A^{[1]}_{k_1} A^{[2]}_{k_2} \dots A^{[n]}_{k_N} \right) \big|{k_1, \dots, k_N}\big\rangle,
\]
where $\Tr$ denotes the trace operation.

Defining the hidden and observable Hilbert spaces as $\mathcal{H} = \mathcal{K} = \mathbb{C}^2$ with the standard basis:
\[
e_0 = |0\rangle = \begin{pmatrix} 1 \\ 0 \end{pmatrix}, \quad e_1 =  |1\rangle = \begin{pmatrix} 0 \\ 1 \end{pmatrix},
\]

the corresponding translation-invariant  EHMM associated with the GHZ state, as per Theorem~\ref{thm_main1}, is governed by the hidden partial isometry $V_{H}: e_i \mapsto e_i\otimes e_i, \; i=0,1$ and the observation operator $V_O: e_i \mapsto e_i \otimes |i\rangle$.

\begin{figure}[h!]
\centering
\begin{tikzpicture}[node distance=1.5cm, every node/.style={scale=1.3}]

  \node[draw, circle, fill=blue!20, text=black, thick] (H0) at (-3,0) {\(e_0\)};
  \node[draw, circle, fill=blue!20, text=black, thick] (H1) at (3,0) {\(e_1\)};

  \node[draw, circle, fill=green!20, text=black, thick] (O0) at (-3,-3) {\(|0\rangle\)};
  \node[draw, circle, fill=green!20, text=black, thick] (O1) at (3,-3) {\(|1\rangle\)};

  \draw[->, thick, blue] (H0) to[loop left] node[midway, left, blue] {$1$} (H0);
  \draw[->, thick, blue] (H1) to[loop right] node[midway, right, blue] {$1$} (H1);

  \draw[->, thick, orange] (H0) -- node[midway, left, orange] {$1$} (O0);
  \draw[->, thick, orange] (H1) -- node[midway, right, orange] {$1$} (O1);

\end{tikzpicture}
\caption{Hidden Markov Model (HMM) for the GHZ state with an identity transition matrix \( \Pi = I_2 \) and identity observation matrix \( Q = I_2 \). Each hidden state remains unchanged (\( e_0 \) stays as \( e_0 \), and \( e_1 \) stays as \( e_1 \)), and each deterministically maps to its corresponding observable state (\( e_0 \rightarrow |0\rangle \), \( e_1 \rightarrow |1\rangle \)).}
\label{fig:ghz_hmm}
\end{figure}
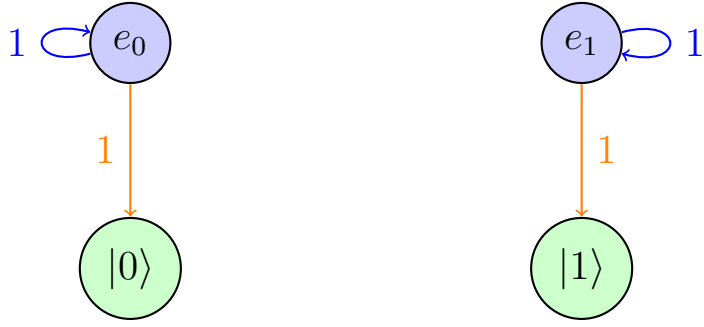

Consequently, the EHMM is generated by the vector:
\begin{equation}\label{eq-HMM-GHZ}
 \big|\Psi_{H,O;n}\big\rangle = \sum_{i}\frac{1}{\sqrt{2}}\, e_i^{\otimes(n+1)}\otimes |i\rangle^{\otimes n},
\end{equation}
for an initial probability distribution given by $\pi = \begin{pmatrix} \frac{1}{2}  &  \frac{1}{2} \end{pmatrix}$.

Notably, the MPS formulation of the GHZ state satisfies the gauge condition \eqref{eq_sumAkAkd1}, and the resulting EHMM aligns precisely with \eqref{eq-HMM-GHZ}, demonstrating a clear correspondence between the GHZ state and the EHMM framework as characterized by Theorem \ref{thm_main2}.

\subsection{Cluster State Representation}

The one-dimensional (1D) cluster state is the unique ground state of the stabilizer Hamiltonian:
\begin{equation}
    H = -\sum_{i=1}^{N-1} \sigma^z_{i-1} \sigma^x_i \sigma^z_{i+1},
\end{equation}
where the Pauli matrices are explicitly given by:
\begin{equation}
    \sigma^x = \begin{bmatrix} 0 & 1 \\ 1 & 0 \end{bmatrix}, \quad
    \sigma^z = \begin{bmatrix} 1 & 0 \\ 0 & -1 \end{bmatrix}.
\end{equation}
We define both the hidden and observable Hilbert spaces as $\mathcal{H} = \mathcal{K} = \mathbb{C}^2$, with the standard computational basis:
\[
e_0 = |0\rangle = \begin{pmatrix} 1 \\ 0 \end{pmatrix}, \quad
e_1 = |1\rangle = \begin{pmatrix} 0 \\ 1 \end{pmatrix}.
\]

The cluster state can be described in the MPS formalism with the following site-independent tensors:
\begin{equation}
    A_0^{[n]} = \frac{1}{\sqrt{2}} \begin{bmatrix} 1 & 1 \\ 0 & 0 \end{bmatrix}, \quad
    A_1^{[n]} = \frac{1}{\sqrt{2}} \begin{bmatrix} 0 & 0 \\ 1 & -1 \end{bmatrix}.
\end{equation}
This representation satisfies the gauge condition (\ref{eq_SAkA1}), meaning that the map:
\begin{equation}
    M \mapsto A_0^{[n]} M A_0^{[n]} + A_1^{[n]} M A_1^{[n]}
\end{equation}
defines a   Markov operator on the space of $2 \times 2$ matrices, $\mathcal{M}_2(\mathbb{C})$.

One can verify that this tensor representation conforms to the decomposition (\ref{eq_aij}) using the matrices:
\begin{equation}
U =  \frac{1}{\sqrt{2}} \begin{bmatrix} 1 & 1 \\ -1 & 1 \end{bmatrix}, \qquad
\chi = \mathbb{I}_2 = \begin{bmatrix} 1 & 0 \\ 0 & 1 \end{bmatrix}.
\end{equation}

The matrix $U$ in the decomposition is precisely the Hadamard gate,
which plays a crucial role in quantum information theory by mapping computational basis states to equal superpositions. Here, it governs the transformation of the hidden Markov process.

\begin{figure}[h!]
\centering
\begin{tikzpicture}[node distance=1.5cm, every node/.style={scale=1.3}]

  \node[draw, circle, fill=blue!20, text=black, thick] (H0) at (-3,0) {\(e_0\)};
  \node[draw, circle, fill=blue!20, text=black, thick] (H1) at (3,0) {\(e_1\)};

  \node[draw, circle, fill=green!20, text=black, thick] (O0) at (-3,-3) {\(|0\rangle\)};
  \node[draw, circle, fill=green!20, text=black, thick] (O1) at (3,-3) {\(|1\rangle\)};

  \draw[->, thick, blue] (H0) to[bend left=30] node[midway, above, blue] {$\frac{1}{2}$} (H1);
  \draw[->, thick, blue] (H1) to[bend left=30] node[midway, below, blue] {$\frac{1}{2}$} (H0);
  \draw[->, thick, blue] (H0) to[loop left] node[midway, left, blue] {$\frac{1}{2}$} (H0);
  \draw[->, thick, blue] (H1) to[loop right] node[midway, right, blue] {$\frac{1}{2}$} (H1);

  \draw[->, thick, orange] (H0) -- node[midway, left, orange] {$1$} (O0);
  \draw[->, thick, orange] (H1) -- node[midway, right, orange] {$1$} (O1);

\end{tikzpicture}
\caption{Hidden Markov Model (HMM) for the 1D cluster state with  hidden transition matrix  $\Pi =  \frac{1}{2} \begin{bmatrix} 1 & 1 \\  1 & 1 \end{bmatrix}$ between the two states \( e_0 \) and \( e_1 \) and an identity emission matrix \( Q = I_2 \). Each hidden state deterministically emits a corresponding observable state, meaning \( e_0 \) always maps to \( |0\rangle \) and \( e_1 \) always maps to \( |1\rangle \). The transitions between hidden states occur with equal probability \( \frac{1}{2} \).}
\label{fig:cluster_hmm_2states}
\end{figure}
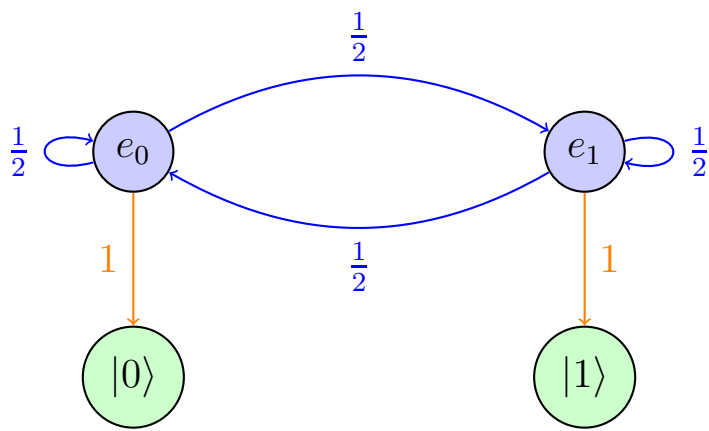

More concretely, the hidden state dynamics are regulated by the Hadamard gate through the partial isometry:
\begin{equation}
    V_H(e_i) = \frac{1}{\sqrt{2}} \big(e_i \otimes e_i + (-1)^i e_i \otimes e_{1-i}\big), \quad i = 0,1.
\end{equation}
This transformation ensures that the hidden process evolves in a manner consistent with the underlying cluster state entanglement structure.

Simultaneously, the observable emission process is governed by another partial isometry $V_O$, which maps:
\begin{equation}
    V_O: e_i \mapsto e_i \otimes |i\rangle.
\end{equation}
This guarantees that the emission probabilities align with the quantum measurement outcomes.

Moreover, the  EHMM  constructed from this framework aligns precisely with the structure described in Theorem \ref{thm_main2}. The interplay between the Hadamard gate and the cluster state structure highlights the non-trivial hidden dynamics, where entanglement and local transformations are encoded through the EHMM formalism.

\subsection{ The AKLT model and its HMM Representation}\label{sect-AKLT}

The AKLT state (Affleck-Kennedy-Lieb-Tasaki) is the  ground state of the AKLT Hamiltonian on a one-dimensional spin-1 chain. This Hamiltonian is defined as:

\[
\hat{H} = \sum_{\langle i,j \rangle} \left( \vec{S}_i \cdot \vec{S}_j + \frac{1}{3} (\vec{S}_i \cdot \vec{S}_j)^2 \right)
\]

where \( \vec{S}_i = (S_i^x, S_i^y, S_i^z) \) are spin-1 operators at site \( i \), and the sum runs over nearest-neighbor pairs \( \langle i,j \rangle \)  satisfying the commutation relations of the \( \mathfrak{su}(2) \) Lie algebra:

\[
[S^x, S^y] = i S^z, \quad [S^y, S^z] = i S^x, \quad [S^z, S^x] = i S^y
\]

 Consider the hidden Hilbert space $\mathcal{H} = \mathbb{C}^2$ with orthonormal basis $e_1 = \begin{bmatrix}
1 \\
 0
\end{bmatrix}, \, e_2 = \begin{bmatrix}
0 \\
 1
\end{bmatrix}$ and the  observation Hilbert space $\mathcal{K} = \mathbb{C}^3$ with  orthonormal spin basis
 \( |+\rangle, |0\rangle, |-\rangle \) correspond to the matrices:
\begin{equation}\label{eq_AKLT}
A^{[n]}_{+} = \sqrt{\frac{2}{3}} \, \sigma^{+}, \quad A^{[n]}_{0} = \sqrt{\frac{1}{3}} \, \sigma^{z}, \quad A^{[n]}_{-} = -\sqrt{\frac{2}{3}} \, \sigma^{-}
\end{equation}
where \( \sigma^{+} \), \( \sigma^{z} \), and \( \sigma^{-} \) are the Pauli matrices:
\[
\sigma^{+} = \begin{pmatrix} 0 & 1 \\ 0 & 0 \end{pmatrix}, \quad \sigma^{z} = \begin{pmatrix} 1 & 0 \\ 0 & -1 \end{pmatrix}, \quad \sigma^{-} = \begin{pmatrix} 0 & 0 \\ 1 & 0 \end{pmatrix}
\]
 The associated state is denoted as:
\[
|\text{AKLT}\rangle_N = \sum_{k_1,\dots,k_n} \text{Tr}\left(A_{k_1}^{[1]} A_{k_2}^{[2]} \cdots A_{k_N}^{[N]}\right) |k_1 k_2 \cdots k_N\rangle,
\]

\begin{proposition}
    The AKLT matrix product state cannot be realized as a partial observation of an  EHMM as described in Theorem \ref{thm_main1}.
\end{proposition}

\begin{proof}
    Assume, for the sake of contradiction, that such a decomposition exists. In this case, the AKLT matrices given by (\ref{eq_AKLT}) would satisfy the representation
    \[
        a_{k;ij} = U_{ij} \chi_i(k),
    \]
    for all \( i, j \in \{1, 2\} \) and \( k \in \{+, 0, -\} \), where:
    \begin{itemize}
        \item \( U = (U_{ij})_{1 \leq i, j \leq 2} \) is a unitary matrix, and
        \item \( \chi = (\chi_i(k))_{\substack{1 \leq i \leq 2 \\ k \in \{+, 0, -\}}} \) is a matrix such that \( Q = \big(|\chi_i(k)|^2\big)_{\substack{1 \leq i \leq 2 \\ k \in \{+, 0, -\}}} \) is stochastic, i.e., all entries of \( Q \) are non-negative, and each row sums to 1.
    \end{itemize}

    Substituting specific values for \( k \), \( i \), and \( j \), we derive the following conditions:
    \[
    \begin{cases}
        U_{11} \chi_1(+) = 0, & \text{for } k = +, \, i = j = 1, \\
        U_{12} \chi_1(+) = \sqrt{\frac{2}{3}}, & \text{for } k = +, \, i = 1, \, j = 2, \\
        U_{11} \chi_1(0) = \sqrt{\frac{1}{3}}, & \text{for } k = 0, \, i = j = 1.
    \end{cases}
    \]

    From the first and the second equations, it follows that \( U_{11} = 0 \). However, this directly contradicts the third equation, which requires \( U_{11} \neq 0 \) to satisfy \( U_{11} \chi_1(0) = \sqrt{\frac{1}{3}} \).

    This contradiction implies that such a decomposition is impossible. Consequently, the AKLT matrices given by (\ref{eq_AKLT}) cannot satisfy the representation in (\ref{eq_aij}). This completes the proof.
\end{proof}
Although the AKLT model cannot be directly derived from an  EHMM under the conditions of Theorem \ref{thm_main1}, the matrices (\ref{eq_AKLT}) satisfy the gauge condition \ref{eq_SAkA1}. Thus, the assumptions of Theorem \ref{thm_main2} hold, allowing the construction of an EHMM associated with the AKLT model, characterized by the homogeneous transition and emission matrices:

\begin{equation}\label{eq_PQ_AKLT}
\Pi = \begin{pmatrix} \frac{1}{3} & \frac{2}{3} \\ \\ \frac{2}{3} & \frac{1}{3} \end{pmatrix}, \quad
Q = \begin{pmatrix} \frac{2}{3} & \frac{1}{3} & 0 \\ \\ 0 & \frac{1}{3} & \frac{2}{3} \end{pmatrix}.
\end{equation}

Here, \(\Pi\) represents the transition probabilities between hidden states, while \(Q\) defines the emission probabilities. The dynamics of this HMM are illustrated in the following diagram:
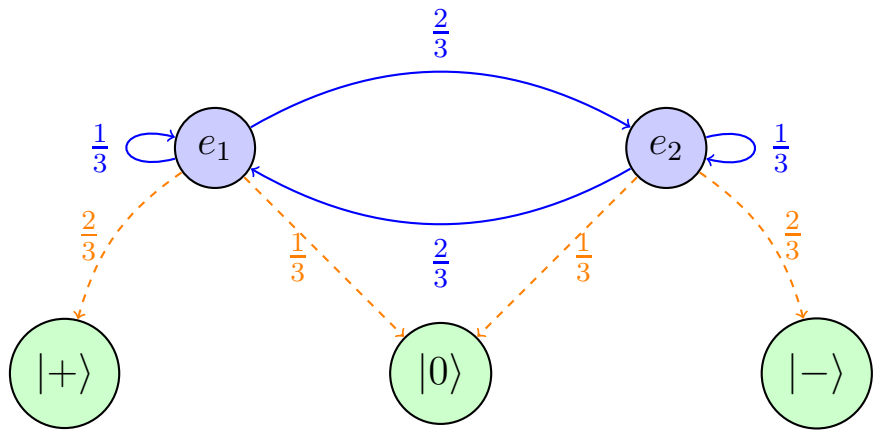
\begin{figure}[h!]
\centering
\begin{tikzpicture}[node distance=1.5cm, every node/.style={scale=1.3}]

  \node[draw, circle, fill=blue!20, text=black, thick] (H1) at (-3,0) {\(e_1\)};
  \node[draw, circle, fill=blue!20, text=black, thick] (H2) at (3,0) {\(e_2\)};

  \node[draw, circle, fill=green!20, text=black, thick] (O1) at (-5,-3) {\(|+\rangle\)};
  \node[draw, circle, fill=green!20, text=black, thick] (O2) at (0,-3) {\(|0\rangle\)};
  \node[draw, circle, fill=green!20, text=black, thick] (O3) at (5,-3) {\(|-\rangle\)};

  \draw[->, thick, blue] (H1) to[bend left=30] node[midway, above, blue] {$\frac{2}{3}$} (H2);
  \draw[->, thick, blue] (H2) to[bend left=30] node[midway, below, blue] {$\frac{2}{3}$} (H1);
  \draw[->, thick, blue] (H1) to[loop left] node[midway, left, blue] {$\frac{1}{3}$} (H1);
  \draw[->, thick, blue] (H2) to[loop right] node[midway, right, blue] {$\frac{1}{3}$} (H2);

  \draw[->, thick, orange, dashed] (H1) to[bend right=20] node[midway, left, orange] {$\frac{2}{3}$} (O1);
  \draw[->, thick, orange, dashed] (H1) -- node[midway, left, orange] {$\frac{1}{3}$} (O2);

  \draw[->, thick, orange, dashed] (H2) -- node[midway, right, orange] {$\frac{1}{3}$} (O2);
  \draw[->, thick, orange, dashed] (H2) to[bend left=20] node[midway, right, orange] {$\frac{2}{3}$} (O3);

\end{tikzpicture}
\caption{This diagram represents the HMM for an AKLT state. The hidden states \( e_1\) and \(e_2\) transition with probabilities as described by the matrix \(\Pi\), shown in blue to match the hidden states. The observations \(|+\rangle\), \(|0\rangle\), and \(|-\rangle\) are emitted with probabilities given by the matrix \(Q\), highlighted in orange for distinction.  }
\label{fig:aklt_hmm_colored}
\end{figure}


We proceed to derive a new MPS originating from the EHMM described by the matrices (\ref{eq_PQ_AKLT}), as outlined in Theorem \ref{thm_main1}. To this end, let us consider the following matrices:

\[
U = \begin{pmatrix}
\sqrt{\frac{1}{3}} & \sqrt{\frac{2}{3}} \\
\\
-\sqrt{\frac{2}{3}} & \sqrt{\frac{1}{3}}
\end{pmatrix}, \quad
\chi = \begin{pmatrix}
\sqrt{\frac{2}{3}} & \sqrt{\frac{1}{3}} & 0 \\
\\
0 & \sqrt{\frac{1}{3}} & -\sqrt{\frac{2}{3}}
\end{pmatrix}.
\]

The matrix \( U \) is orthogonal, and its entries satisfy \( |U_{ij}|^2 = \Pi_{ij} \), while \( |\chi_{i}(k)|^2 = Q_i(k) \) for all indices \( i \), \( j \), and \( k \). Consequently, the relevant matrices corresponding to (\ref{eq_aij}), with entries \( a^{'[n]}_{k; ij} = U^{[n]}_{ij} \chi^{[n]}_{i}(k) \), take the following form:

\begin{equation}\label{eq_newAKLT}
A^{'[n]}_{+} = \frac{1}{3}
\begin{pmatrix}
\sqrt{2} & 2 \\
\\
0 & 0
\end{pmatrix}, \quad
A^{'[n]}_{0} = \frac{1}{3}
\begin{pmatrix}
1 & \sqrt{2} \\
\\
-\sqrt{2} & 1
\end{pmatrix}, \quad
A^{'[n]}_{-} = \frac{1}{3}
\begin{pmatrix}
0 & 0 \\
\\
2 & -\sqrt{2}
\end{pmatrix}.
\end{equation}

A natural question arises as to whether the matrices \( \{A_k^{'[n]}\} \) provide an alternative representation of the AKLT state. However, it is clear that this is not the case. The MPS generated by the newly defined matrices (\ref{eq_newAKLT}) is distinct from the AKLT state. Investigating its relationship with the original state and its connection to other similar states arising from different choices of the matrices \( U \) and \( \chi \) is an intriguing direction for further exploration.

\end{document}